\documentclass[11pt]{article}

\usepackage{graphicx}
\usepackage{amsfonts}
\usepackage{amsmath}
\usepackage{amsthm}
\usepackage[mathcal]{euscript}      % Define tensors
\usepackage{bm}                     % Define vectors
\usepackage{color}
\usepackage{epstopdf}
\usepackage{amssymb}
\usepackage{mathrsfs,mathdots}
\usepackage{lscape}
\usepackage[version=3]{mhchem}

\usepackage[paper=a4paper,dvips,top=3cm,left=2.4cm,right=2.4cm,
    foot=1cm,bottom=3cm]{geometry}

\begin{document}

\title{Two Irreducible Functional Bases of Isotropic Invariants of A Fourth Order Three-Dimensional Symmetric and Traceless Tensor}
\author{Zhongming Chen\footnote{%
    Department of Mathematics, School of Science, Hangzhou Dianzi University, Hangzhou 310018, China ({\tt zmchen@hdu.edu.cn}).
    This author's work was supported by the National Natural Science Foundation of China (Grant No. 11701132).}
\and Yannan Chen\footnote{%
    School of Mathematics and Statistics, Zhengzhou University, Zhengzhou 450001, China ({\tt ynchen@zzu.edu.cn}).
    This author was supported by the National Natural Science Foundation of China (Grant No. 11571178, 11771405).}
\and Liqun Qi\footnote{%
    Department of Applied Mathematics, The Hong Kong Polytechnic University,
    Hung Hom, Kowloon, Hong Kong ({\tt maqilq@polyu.edu.hk}).
    This author's work was partially supported by the Hong Kong Research Grant Council
    (Grant No. PolyU  15302114, 15300715, 15301716 and 15300717).}
\and Wennan Zou\footnote{%
    Institute for Advanced Study, Nanchang University, Nanchang 330031, China ({\tt zouwn@ncu.edu.cn}).
    This author was supported by the National Natural Science Foundation of China (Grant No. 11372124)}
    }

\date{\today}
\maketitle

\begin{abstract}    The elasticity tensor is one of the most important fourth order tensors in mechanics.   Fourth order three-dimensional symmetric and traceless tensors play a crucial role in the study of the elasticity tensors.   In this paper, we present two isotropic irreducible functional bases of a fourth order three-dimensional symmetric and traceless tensor.  One of them is the minimal integrity basis introduced by Smith and Bao in 1997.   It has nine homogeneous polynomial invariants of degrees two, three, four, five, six, seven, eight, nine and ten, respectively.  We prove that it is also an irreducible functional basis.   The second irreducible functional basis also has nine homogeneous polynomial invariants.  It has no quartic invariant but has two sextic invariants.   The other seven invariants are the same as those of the Smith-Bao basis.    Hence, the second irreducible functional basis is not contained in any minimal integrity basis.

  \textbf{Key words.} irreducible functional basis, symmetric and traceless tensor, invariant.
\end{abstract}

\newtheorem{Theorem}{Theorem}[section]
\newtheorem{Definition}[Theorem]{Definition}
\newtheorem{Lemma}[Theorem]{Lemma}
\newtheorem{Corollary}[Theorem]{Corollary}
\newtheorem{Proposition}[Theorem]{Proposition}
\newtheorem{Conjecture}[Theorem]{Conjecture}
\newtheorem{Question}[Theorem]{Question}

% LaTeX definitions
\renewcommand{\hat}[1]{\widehat{#1}}
\renewcommand{\tilde}[1]{\widetilde{#1}}
\renewcommand{\bar}[1]{\overline{#1}}
\newcommand{\REAL}{\mathbb{R}}
\newcommand{\COMPLEX}{\mathbb{C}}
\newcommand{\SPHERE}{\mathbb{S}^2}
\newcommand{\diff}{\,\mathrm{d}}
\newcommand{\st}{\mathrm{s.t.}}
\newcommand{\T}{\top}
\newcommand{\vt}[1]{{\bf #1}}%{\bm{#1}}
\newcommand{\aaa}{{\vt{a}}}
\newcommand{\ddd}{{\vt{d}}}
\newcommand{\x}{{\vt{x}}}
\newcommand{\y}{{\vt{y}}}
\newcommand{\z}{{\vt{z}}}
\newcommand{\uu}{{\vt{u}}}
\newcommand{\vv}{{\vt{v}}}
\newcommand{\ww}{{\vt{w}}}
\newcommand{\e}{{\vt{e}}}
\newcommand{\g}{{\vt{g}}}
\newcommand{\0}{{\vt{0}}}
\newcommand{\Ten}{\bf{T}}
\newcommand{\HH}{\mathbb{H}}
\newcommand{\A}{{\bf A}}
\newcommand{\B}{{\bf B}}
\newcommand{\C}{\mathcal{C}}
\newcommand{\D}{{\bf D}}
\newcommand{\E}{\bf{E}}
\newcommand{\OOO}{\mathcal{O}}
\newcommand{\U}{{\bf{U}}}
\newcommand{\V}{{\bf{V}}}
\newcommand{\W}{{\bf{W}}}
\newcommand{\I}{{\bf{I}}}
\newcommand{\II}{{\mathcal{I}}}
\newcommand{\OO}{{\bf{O}}}
\newcommand{\RESULTANT}{\mathrm{Res}}

\newpage
\section{Introduction}

The elasticity tensor is one of the most important fourth order tensors in mechanics.   There are many papers to study properties of the elasticity tensor.   One of the important theoretical topics for the elasticity tensor in a three dimensional physical space is to study minimal integrity bases and irreducible functional bases of its isotropic invariants.    In 1994, Boehler, Kirillov and Onat \cite{BKO-94} studied
the polynomial basis of anisotropic invariants of the elasticity tensor.   In 2017, Olive, Kolev and Auffray \cite{OKA-17} presented a minimal integrity basis of isotropic invariants of the elasticity tensor, with $297$ invariants.    It is well-known that the number of invariants with the same degree in a minimal integrity basis of some tensors is always fixed \cite{Spe-71}.  It is also possible that the cardinality of an irreducible functional basis of a certain tensor is smaller than the cardinality of a minimal integrity basis of that tensor.  For example, in 2014, Olive and Auffray \cite{OA-14} presented a minimal integrity basis of $13$ isotropic invariants for a third order three-dimensional symmetric tensor.   Recently, Chen, Liu, Qi, Zheng and Zou \cite{CLQZZ-18} presented an irreducible functional basis of that tensor, with $11$ polynomial invariants.    Thus, it is very possible that the elasticity tensor may have a functional basis consisting of polynomial invariants, whose cardinality is smaller than $297$.   If such a functional basis of
the elasticity tensor can be found and its cardinality is significantly smaller than $297$, then it will significant to both theoretical and applied mechanics.

Until now, there are no such a result for the elasticity tensor.  This topic on irreducible functional bases of the elasticity tensor may not be very easy.
On the other hand, the elasticity tensor in a three dimensional physical space has an orthogonal irreducible decomposition \cite{ZZDR-01}, which contains five parts: a fourth order symmetric and traceless tensor, two second order symmetric and traceless tensors, and two scalars. Clearly, scalars are naturally isotropic invariants. Irreducible functional bases of isotropic invariants of second order symmetric and traceless tensors are also well-known \cite{Zh-94}. However, irreducible functional bases of isotropic invariants of a fourth order symmetric and traceless tensor and irreducible functional bases of joint isotropic invariants of a fourth order symmetric and traceless tensor and second order symmetric and traceless tensors are still open. To begin with, we may study irreducible functional bases of a fourth order three-dimensional symmetric and traceless tensor.   %We know that fourth order three-dimensional symmetric and traceless tensors play a crucial role in the study of the elasticity tensors.    
In the study of the elasticity tensor, Boehler, Kirillov and Onat \cite{BKO-94} presented a minimal integrity basis for a fourth order three-dimensional symmetric and traceless tensor with nine homogeneous polynomial invariants.   These nine homogeneous polynomial invariants have degrees two, three, four, five, six, seven, eight, nine and ten, respectively.  In 1997, Smith and Bao \cite{SB-97} presented another minimal integrity basis for a fourth order three-dimensional symmetric and traceless tensor.   The sextic invariant of the Smith and Bao basis is different from that of the Boehler, Kirillov and Onat basis.   The other eight invariants of the Smith and Bao basis are the same as the Boehler, Kirillov and Onat basis.   On the other hand, we may try to find some irreducible functional bases of a fourth order three-dimensional symmetric and traceless tensor.

In this paper, we present two isotropic irreducible functional bases of a fourth order three-dimensional symmetric and traceless tensor.  One of them is exactly the Smith and Bao basis.   We prove that it is also an irreducible functional basis.   The second irreducible functional basis also has nine homogeneous polynomial invariants.  It has no quartic invariant but has two sextic invariants.   One of the two sextic invariants is the same as the one in the Boehler, Kirillov and Onat basis.   Another sextic invariant is the same as the one in the Smith and Bao basis.  The other seven invariants are the same as those of the Smith-Bao basis and the Boehler, Kirillov and Onat basis.    Hence, the second irreducible functional basis is not contained in any minimal integrity basis.
This shows that for a fourth order three-dimensional symmetric and traceless tensor, there is such an irreducible functional basis which consists of polynomial invariants, but not contained in any minimal integrity basis.

In the next section, some preliminary results are given. The Smith and Bao basis is described there.   Then we introduce the other basis described above and show that it is a functional basis in Section 3.   After this we will call this functional basis the mixed functional basis.  We use the divide and conquer approach to deal with these two bases.  The Smith and Bao basis and the mixed functional basis have ten different homogeneous polynomial invariants.  We divide these ten invariants to two groups.   One group consists of four invariants of odd degrees.   The other group consists of invariants of even degrees.  In Section 4, we show that each of these four odd degree invariants is not a function of the other nine invariants.   Then
we show that the Smith and Bao basis is an irreducible functional basis in Section 5, the mixed functional basis is also an irreducible functional basis in Section 6.   The main tactics used in Sections 5 and 6 is to keep the four odd degree invariants zero by restricting five of the nine independent elements of the fourth order three-dimensional symmetric and traceless tensor to zero.   This tactics reduces the size of the systems to be solved.

\section{Preliminaries}

In this paper, we consider the three-dimensional physical space.  The summation convention is used.    If an index is repeated twice in a product, then it means that this product is summed up with respect to this index from $1$ to $3$.

%Suppose that a tensor $\A$ has the form $A_{i_1\dots i_m}$ under an orthonormal basis $\{ \e_i \}$.  A scalar function of $\A$, $f(\A) = f(A_{i_1\dots i_m})$ is said to be an isotropic invariant of $\A$ if for any orthogonal matrix $q_{ij}$, we have

From now on, we use $\A$ to denote a fourth order three-dimensional tensor, and assume that it is represented by $A_{ijkl}$ under an orthonormal basis $\{ \e_i \}$.
  Hence $i, j, k, l \in \{1, 2, 3\}$.   We say that $\A$ is a symmetric tensor if $A_{ijkl}$ is invariant under any permutation of indices.
We say that $\A$ is traceless if
$$A_{iijk} = A_{ijik} = A_{ijki} =  A_{jiik} = A_{jiki} = A_{jkii} = 0 \qquad \forall j,k.$$
We use $\D$ to denote a fourth order three-dimensional symmetric and traceless tensor.
We use $\0$ to denote the zero vector and $\OOO$ to denote the third order three-dimensional zero tensor.

 Consider a fourth order three-dimensional tensor $\A$.  A scalar function of $\A$, $f(\A) = f(A_{ijkl})$ is said to be an isotropic invariant of $\A$ if for any orthogonal matrix $Q_{ii'}$, we have
$$f(A_{ijkl}) = f(A_{i'j'k'l'}Q_{ii'}Q_{jj'}Q_{kk'}Q_{ll'}).$$

   A set of isotropic polynomial invariants $f_1, \dots, f_r$ of $\A$ is said to be an \emph{integrity basis} of $\A$ if any isotropic polynomial invariant is a polynomial of $f_1, \dots, f_r$, and a set of isotropic invariants $f_1, \dots, f_m$ of $\A$ is said to be a \emph{functional basis} of $\A$ if any isotropic invariant is a function of $f_1, \dots, f_m$.  An integrity basis is always a functional basis but not vice versa \cite{BKO-94}.    A set of isotropic polynomial invariants $f_1, \dots, f_r$ of $\A$ is said to be \emph{polynomially irreducible} if none of them is a polynomial of the others. Similarly, a set of isotropic invariants $f_1, \dots, f_m$ of $\A$ is said to be \emph{functionally irreducible} if none of them is a function of the others.   A polynomially irreducible integrity basis of $\A$ is said to be a \emph{minimal integrity basis} of $\A$.  A functionally irreducible functional basis of $\A$ is said to be an \emph{irreducible functional basis} of $\A$.

Consider a fourth order three-dimensional symmetric and traceless tensor $\D$. We use notation: $B_{ij}=D_{ik\ell m}D_{jk\ell m}$, $B^2_{ij}=B_{ik}B_{kj}$ and $C_{ijk\ell}=D_{ijmn}D_{k\ell mn}$.

We observe that there are nine independent elements in $\D$:
\begin{equation*}
  D_{1111}, D_{1112}, D_{1113}, D_{1122}, D_{1123}, D_{1222}, D_{1223}, D_{2222}, \text{ and } D_{2223}.
\end{equation*}

The Smith and Bao minimal integrity basis \cite{SB-97} of a fourth order three-dimensional symmetric and traceless tensor $\D$ is given by the following nine invariants with degrees $2,\dots,10$:
  \begin{equation*}
    \begin{array}{lll}
      J_2:=D_{ijk\ell}D_{ijk\ell},         & J_3:=C_{ijk\ell}D_{ijk\ell},          & J_4:=B_{ij}B_{ij}, \\
      J_5:=B_{ij}D_{ijk\ell}B_{k\ell},     & J_6:=B_{ij}C_{ijk\ell}B_{k\ell},         & J_7:=B^2_{ij}D_{ijk\ell}B_{k\ell}, \\
      J_8:=B^2_{ij}C_{ijk\ell}B_{k\ell},   & J_9:=B^2_{ij}D_{ijk\ell}B^2_{k\ell},  & J_{10}:=B^2_{ij}C_{ijk\ell}B^2_{k\ell}.
    \end{array}
  \end{equation*}

\section{The Mixed Functional Basis}

The Boehler, Kirillov and Onat minimal integrity basis \cite{BKO-94} of a fourth order three-dimensional symmetric and traceless tensor $\D$ is the same as the Smith and Bao minimal integrity basis except $J_6$ is replaced by
$$K_6 = B_{ij} B_{jk} B_{ki}.$$

We have following theorem.

\begin{Theorem} \label{t3.1}
 The mixed basis $\{J_2,J_3,J_5,J_6,K_6,J_7,J_8,J_9,J_{10}\}$ is a functional basis of the fourth order three-dimensional symmetric and traceless tensor $\D$.
\end{Theorem}
\begin{proof}
Since the Smith and Bao basis $\{J_2,J_3,J_4,J_5,J_6,J_7,J_8,J_9,J_{10}\}$ is an integrity basis of $\D$, it is also a functional basis of $\D$ \cite{BKO-94}.  Thus, by adding an invariant, $K_6$, the ten invariant set $\{J_2,J_3,J_4,J_5,J_6,K_6,J_7,J_8,J_9,J_{10}\}$ is also a functional basis of $\D$.    Since the Smith and Bao basis $\{J_2,J_3,J_4,J_5,J_6,J_7,J_8,J_9,J_{10}\}$ is an integrity basis of $\D$, $K_6$ should be a linear combination of $J_2^3$, $J_2J_4$, $J_3^2$ and $J_6$.    We have
$$K_6 = -\frac{13}{80} J_2^3 + \frac{33}{40} J_2 J_4 - \frac{1}{24} J_3^2 + \frac{9}{16} J_6.$$
Then
$$J_2J_4 = \frac{39 J_2^3 + 10J_3^2 -135 J_6 +240 K_6}{198}.$$
If $J_2 \not = 0$, we have
$$J_4 = \frac{39 J_2^3 + 10J_3^2 -135 J_6 +240 K_6}{198 J_2}.$$
If $J_2 = 0$, then $\D = 0$ and we have $J_4 = 0$.  Hence, $J_4$ is a function of $J_2, J_3, J_6$ and $K_6$.   Therefore, the nine invariant set $\{J_2,J_3,J_5,J_6,K_6,J_7,J_8,J_9,J_{10}\}$ is a functional basis of $\D$.
\end{proof}

From now on, we call this nine invariant set $\{J_2,J_3,J_5,J_6,K_6,J_7,J_8,J_9,J_{10}\}$ the mixed functional basis of $\D$.   Since it has no quartic invariant and has two sextic invariants, it is not contained in any minimal integrity basis of $\D$.

\section{The Odd Degree Invariants in These Two Bases}

In the Smith and Bao basis and the mixed functional basis, there are four odd degree invariants $J_3, J_5, J_7$ and $J_9$, and six even degree invariants $J_2, J_4, J_6, K_6, J_8$ and $J_{10}$. A notable property is that when $\D$ changes its sign, the odd degree invariants change their signs, while the even degree invariants are unchanged.   Using this property, it is relatively easy to show that each of $J_3, J_5, J_7$ and $J_9$ is not a function of the other three odd degree invariants and the six even degree invariants $J_2, J_4, J_6, K_6, J_8$ and $J_{10}$.

\begin{Proposition} \label{p4.1}  We have the following four conclusions.

(a) If there is a fourth order three-dimensional symmetric and traceless tensor $\D$ such that $J_5 = J_7 = J_9 = 0$ but $J_3 \not = 0$, then $J_3$ is not a function of $J_2, J_4, J_5, J_6, K_6, J_7, J_8, J_9$ and $J_{10}$.

(b) If there is a fourth order three-dimensional symmetric and traceless tensor $\D$ such that $J_3 = J_7 = J_9 = 0$ but $J_5 \not = 0$, then $J_5$ is not a function of $J_2, J_3, J_4, J_6, K_6, J_7, J_8, J_9$ and $J_{10}$.

(c) If there is a fourth order three-dimensional symmetric and traceless tensor $\D$ such that $J_3 = J_5 = J_9 = 0$ but $J_7 \not = 0$, then $J_7$ is not a function of $J_2, J_3, J_4, J_5, J_6, K_6, J_8, J_9$ and $J_{10}$.

(d) If there is a fourth order three-dimensional symmetric and traceless tensor $\D$ such that $J_3 = J_5 = J_7 = 0$ but $J_9 \not = 0$, then $J_9$ is not a function of $J_2, J_3, J_4, J_5, J_6, K_6, J_7, J_8$ and $J_{10}$.
 \end{Proposition}
\begin{proof}
We now prove conclusion (a).  If there is a fourth order three-dimensional symmetric and traceless tensor $\D$ such that $J_5 = J_7 = J_9 = 0$ but $J_3 \not = 0$, then we may consider $-\D$. $J_2, J_4, J_5, J_6, K_6$, $J_7, J_8, J_9$ and $J_{10}$ are unchanged,  but $J_3$ changes its sign.   This implies that $J_3$ is not a function of $J_2, J_4, J_5, J_6, K_6$, $J_7, J_8, J_9$ and $J_{10}$.    The other three conclusions (b), (c) and (d) can be proved similarly.
\end{proof}

We now have the following theorem.

\begin{Theorem} \label{t4.2}
Each of $J_3, J_5, J_7$ and $J_9$ is not a function of the other three odd degree invariants and the six even degree invariants $J_2, J_4, J_6, K_6, J_8$ and $J_{10}$.
\end{Theorem}
\begin{proof}
To show that each of $J_3, J_5, J_7$ and $J_9$ is not a function of the other nine invariants, by Proposition \ref{p4.1}, we need to find examples such that one of $J_3, J_5, J_7$ and $J_9$ is not equal to zero, and the other three odd degree invariants vanish.

Let
$$
D_{1111}=8,\   D_{1112}=0, \   D_{1113}=0, \   D_{1122}=-4, \   D_{1123}=0,
$$
$$
D_{1222}=5,  \   D_{1223}=5, \   D_{2222}=3  \text{ \  and \ }  D_{2223}=0.
$$
Then we have $J_5 =J_7=J_9= 0$, $J_3 = -6480$. By part (a) of Proposition \ref{p4.1}, $J_3$ is not a function of $J_2, J_4, J_5, J_6, K_6, J_7, J_8, J_9$ and $J_{10}$.

Let
$$
D_{1111}=1,\   D_{1112}=\sqrt{\frac{5}{2} + \sqrt{\frac{38}{3}}}, \   D_{1113}=-\frac{\sqrt{2}}{2}, \   D_{1122}=-1, \   D_{1123}=0,
$$
$$
D_{1222}=-\sqrt{\frac{5}{2} + \sqrt{\frac{38}{3}}},  \   D_{1223}=0, \   D_{2222}=\frac{1}{2}  \text{ \  and \ }  D_{2223}=\frac{1}{2} \sqrt{5 + 2 \sqrt{\frac{38}{3}}}.
$$
Then we have $J_3 =J_7=J_9= 0$, $J_5 = \frac{25}{2}$. By part (b) of Proposition \ref{p4.1}, $J_5$ is not a function of $J_2, J_3, J_4, J_6, K_6, J_7, J_8, J_9$ and $J_{10}$.

Let
$$
D_{1111}=\frac{1}{144} \left(-83 + \sqrt{13033}\right), \   D_{1112}=0, \   D_{1113}=\frac{1}{72} \sqrt{\frac{1}{2} \left(-11455 + 101 \sqrt{13033}\right)},
$$
$$
D_{1122}=0, \   D_{1123}=0,  \ D_{1222}=\frac{1}{2 \sqrt{2}},  \   D_{1223}=0, \   D_{2222}=\frac{1}{2}  \text{ \  and \ }  D_{2223}=0.
$$
Then we have $J_3 =J_5=J_9= 0$, $J_7 = \frac{6384263 - 55933 \sqrt{13033}}{884736} \thickapprox -0.00132174$. By part (c) of Proposition \ref{p4.1}, $J_7$ is not a function of $J_2, J_3, J_4, J_5, J_6, K_6, J_8, J_9$ and $J_{10}$.

Let
$$
D_{1111}=0,\   D_{1112}=1, \   D_{1113}=0, \   D_{1122}=0, \   D_{1123}=0,
$$
$$
D_{1222}=-\frac{3}{4},  \   D_{1223}=\frac{1}{4}, \   D_{2222}=1  \text{  and }  D_{2223}=0.
$$
Then we have $J_3 =J_5=J_7= 0$, $J_9 = \frac{45}{8}$. By part (d) of Proposition \ref{p4.1}, $J_9$ is not a function of $J_2, J_3, J_4, J_5, J_6, K_6, J_7, J_8$ and $J_{10}$.

\end{proof}

\section{The Smith and Bao Basis is an Irreducible Functional Basis}

 By \cite{BKO-94}, any integrity basis is a functional basis.    Hence, the Smith and Bao basis is a functional basis of $\D$.  To show that the Smith and Bao basis is an irreducible function
 basis of $\D$, we only need to show that each of its nine invariants is not a function of the other eight invariants \cite{PT-87}.

We now present one of the main theorems of this paper.  In the proof of this theorem, we need to show that each of the even degree invariants $J_2, J_4, J_6, J_8$ and $J_{10}$ is not a function of the other eight invariants.   For $J_2$ and $J_4$, we use some simple examples.  These simple examples do not follow the tactics described below.  For $J_{10}$, $J_8$ and $J_6$, we could not find such simple examples.  Then we use the tactics mentioned at the end of the introduction.   We observe that
$$J_3(\D)=J_5(\D)=J_7(\D)=J_9(\D)=0$$
if
$D_{1111}=D_{1112}= D_{1122}=D_{1222}=D_{2222}=0$.
Then we restrict these five independent elements of $\D$ to be zero.   This reduces the size of the systems to be solved.  For $J_{10}$ and $J_8$, we do not give the systems to be solved explicitly, as the solutions of these systems are not too complicated.   For $J_6$, we present the system to be solved explicitly, as the solution of the system is somewhat complicated.

\begin{Theorem} \label{t5.1}
 The Smith and Bao basis $\{J_2,J_3,J_4,J_5,J_6,J_7,J_8,J_9,J_{10}\}$ is an irreducible functional basis of the fourth order three-dimensional symmetric and traceless tensor $\D$.
\end{Theorem}
\begin{proof}

By Theorem \ref{t4.2}, we know that each of $J_3, J_5, J_7$ and $J_9$ is not a function of the other eight invariants.   Hence, it suffices to prove that each of $J_2, J_4, J_6, J_8$ and $J_{10}$ is not a function of the other eight invariants.

We consider three tensors, whose independent elements and associated values of invariants
are listed as
$$
\begin{tabular}{r|ccccccccc}
$(D_{1111}, \ldots, D_{2223})$            & $J_2$ & $J_3$ & $J_4$ &  $J_5$  & $J_6$  & $J_7$  & $J_8$  & $J_9$  & $J_{10}$ \\ \hline
$\D_1 :(1, 0, 0, 0, 0, 0, 0, 0, 0)$        & 8 & 0 &32 &0 &0 &0 &0 &0 &0    \\
$\D_2 :(0,\frac{\sqrt{2}}{\sqrt[4]{11}}, 0, 0, 0, 0, 0, 0, 0)$ & $\frac{32}{\sqrt{11}}$ &0 &32 &0 &0 &0 &0 &0 &0   \\
$\D_3 :(0,\frac{1}{\sqrt{2}}, 0, 0, 0, 0, 0, 0, 0)$    & 8 &0 &22 &0 &0 &0 &0 &0 &0
\end{tabular}
$$
In the table, $(D_{1111}, \ldots, D_{2223})$ means
\begin{equation*}
  D_{1111}, D_{1112}, D_{1113}, D_{1122}, D_{1123}, D_{1222}, D_{1223}, D_{2222}, \text{ and } D_{2223}.
\end{equation*}
We now prove that the invariant $J_2$ is not a function of other invariants by contradiction [10].
Suppose that $J_2$ is a function of $J_3, \ldots , J_{10}$. Once the values of $J_3, \ldots , J_{10}$ are fixed, the value of
$J_2$ is determined. Now, we consider tensors $\D_1$ and $\D_2$. Corresponding values of $J_3, \ldots , J_{10}$ in $\D_1$
and $\D_2$ are equal. But values $J_2(\D_1)$ and $J_2(\D_2)$ are different. This generates the contradiction.
Hence, the invariant $J_2$ must be included in any irreducible functional basis contained in the Smith-Bao integrity basis $\{J_2, \ldots, J_{10}\}$.
Similarly, we can prove that the invariant $J_4$ must be included in any irreducible function
basis contained in the Smith-Bao integrity basis $\{J_2, \ldots, J_{10}\}$, by using tensors $\D_1$ and $\D_3$.

Let $D_{1111}=D_{1112}= D_{1122}=D_{1222}=D_{2222}=0$,
$$
D_{1113}=\frac{1}{9}, D_{1123}=-\frac{1}{9 \sqrt{5}},
D_{1223}=\frac{1}{16} \left(-\frac{4}{9} + \frac{28}{\sqrt{5}}\right), \   \text{  and }  D_{2223}=\frac{\sqrt{5}}{18}.
$$
Then we have
$$ J_3 =J_5=J_7=J_9=0, \ J_2=10, \  J_4=\frac{1553}{45}, $$
$$ J_6=\frac{98}{135}, \   J_8=\frac{207319}{40500}, \  J_{10}=\frac{343 (512675 + 216 \sqrt{5})}{4860000}. $$
Now let $D_{1223}=-\frac{1}{16} \left(\frac{4}{9} + \frac{28}{\sqrt{5}}\right)$ and the values of other  independent elements remain unchanged.
We find that $J_{10} = \frac{343 (512675 - 216 \sqrt{5})}{4860000}$ while the other invariants do not change. This shows that $J_{10}$ is not a function of $J_2, J_3, J_4, J_5, J_6, J_7, J_8$ and $J_{9}$.

Let $\D=(D_{ijkl})$ be a fourth order three-dimensional symmetric and traceless tensor given by $D_{1111}=D_{1112}= D_{1122}=D_{1222}=D_{2222}=0$,
$
D_{1113}=1, \  D_{1123}=-\sqrt{t},
$
$$
D_{1223}=\frac{1}{-4+8t} \left(1-2t + \sqrt{6+9t-54t^2+24t^3} \right),  \text{  and }  D_{2223}=\frac{1+5t}{4\sqrt{t}},
$$
and let $\hat{\D}=(\hat{D}_{ijkl})$ be the same tensor with $\D$ except the element
$$
\hat{D}_{1223}=\frac{1}{-4+8t} \left(1-2t - \sqrt{6+9t-54t^2+24t^3} \right),
$$
where $t \in (0, \frac{1}{2} )$ is a parameter to be determined. By simple computation, we have
$$
J_3(\D) = J_5(\D) = J_7(\D) = J_9(\D) = 0,
$$
$$
J_3(\hat{\D}) =J_5(\hat{\D})= J_7(\hat{\D})=J_9(\hat{\D})=0,
$$
$$
J_2(\D) = J_2(\hat{\D}), \quad   J_4(\D) = J_4(\hat{\D}), \quad J_6(\D) = J_6(\hat{\D}).
$$
On the other hand, we also find that $J_8(\D) = J_8(\hat{\D})$ if and only if
$ t =0.2 $,
and $J_{10}(\D) = J_{10}(\hat{\D})$ if and only if
$$
(1 - 5 t)^2  (-4 - 156 t - 207 t^2 + 5863 t^3 + 6234 t^4 - 24147 t^5 +  9800 t^6) = 0 .
$$
Denote by $h(t)=-4 - 156 t - 207 t^2 + 5863 t^3 + 6234 t^4 - 24147 t^5 +  9800 t^6$. It is easy to check that
$h(0.15)=-10.8359$ and $h(0.2)=6.29856$. Since $h(t)$ is continuous, there exists $t^* \in (0.15, 0.2)$ such that $h(t^*)=0$.
This means that when $t=t^*$, the two tensors $\D$ and $\hat{\D}$ have the same values of $J_2, J_3, J_4, J_5, J_6, J_7, J_9, J_{10}$ except $J_8$.
Hence, $J_8$ is not a function of $J_2, J_3, J_4, J_5, J_6, J_7, J_9, J_{10}$.

Now we turn to show that $J_6$ should be included in the irreducible functional basis. As said before, we observe that
$$J_3(\D)=J_5(\D)=J_7(\D)=J_9(\D)=0$$
if
$D_{1111}=D_{1112}= D_{1122}=D_{1222}=D_{2222}=0$.
In this case, it suffices to find a solution of the homogenous system of equations
$$ \left\{
\begin{array}{rcl}
J_2(D_{1123}, D_{1113}, D_{1223}, D_{2223}) &=&J_2(D_{1123}, D_{1113}, \hat{D}_{1223}, D_{2223}) \\
J_4(D_{1123}, D_{1113}, D_{1223}, D_{2223}) &=&J_4(D_{1123}, D_{1113}, \hat{D}_{1223}, D_{2223}) \\
J_8(D_{1123}, D_{1113}, D_{1223}, D_{2223}) &=&J_8(D_{1123}, D_{1113}, \hat{D}_{1223}, D_{2223}) \\
J_{10}(D_{1123}, D_{1113}, D_{1223}, D_{2223}) &=&J_{10}(D_{1123}, D_{1113}, \hat{D}_{1223}, D_{2223})
\end{array}
\right.
$$
such that $J_6(D_{1123}, D_{1113}, D_{1223}, D_{2223}) \neq J_6(D_{1123}, D_{1113}, \hat{D}_{1223}, D_{2223})$.
Without loss of generality, we set $D_{1123}=1$. By using Wolfram Mathematica 11, we get an exact solution with
$$
D_{1123}=1, \quad D_{1113} =-\sqrt{s}, $$
$$D_{1223}= -\frac{1}{4} + {\sqrt{w} \over v}, \quad
D_{2223}= -\frac{\sqrt{su}}{129040}, \quad
\hat{D}_{1223}=-\frac{1}{4} - {\sqrt{w} \over v},$$
where
$$u = 3409988 + 121090461 s + 3047771035 s^2 + 36002697784 s^3 + 266820285024 s^4 $$ $$ - 258331866402 s^5 - 16053102565014 s^6 -
     57863015604468 s^7 + 67272897295672 s^8 + 488515509520597 s^9 $$ $$ +  183998668490763 s^{10} - 580058593972760 s^{11} + 125410376560000 s^{12},$$
$$w = 162428354809686870437242477762501577487020295460  $$ $$ +          8628548774419344838383115813832547548275782418908 s $$ $$ + 125423295516805616467490217344773033930710687268713 s^2 $$ $$ +          1126867247954259151777196475559771279188654539984224 s^3 $$ $$ + 2533420089208640946542261362251593116475202759287712 s^4 $$ $$ -           46697241050092754413667817534873314111694780885869352 s^5 $$ $$ - 320813241173678850651381840629259210283076735948127042 s^6 $$ $$ -           332968212886054159982148032334380094110168127761229720 s^7 $$ $$ + 2078991876053044927402466766075157597696934830709243244 s^8 $$ $$ +           4819019069384851940098815445425281271098298476930610988 s^9 $$ $$ - 259100661739478059764070674111215079478898258353817655 s^{10} $$ $$ -           4566644851908588078748035381874622354177560283807382600 s^{11} $$ $$ + 1072334324216249440779599886720797670570605466955600000 s^{12}$$ $$ v = {909792 \sqrt{277963020342792405455068513935695065}},
$$
$s \in (0.15, 0.18)$ is the unique solution of the polynomial equation
$ 32260+3571288 x+121090461 x^2+3047771035 x^3+36002697784 x^4+266820285024 x^5-258331866402 x^6-16053102565014 x^7-57863015604468 x^8+67272897295672 x^9+488515509520597 x^{10}+183998668490763 x^{11}-580058593972760 x^{12}+125410376560000 x^{13}=0 $.
Except that $D_{1123}=1$, approximate digit values of $D_{1113}, D_{1223}, D_{2223}, \hat{D}_{1223}$ are as follows:
$$
D_{1113}=-0.406303, \  D_{1223}=0.672665, \  D_{2223}=1.12318  \text{  and }  \hat{D}_{1223}=1.17267.
$$
It is easy to check that $J_6(D_{1123}, D_{1113}, D_{1223}, D_{2223}) \neq J_6(D_{1123}, D_{1113}, \hat{D}_{1223}, D_{2223})$. Thus, $J_6$ is not a function of $J_2, J_3, J_4, J_5, J_7, J_8, J_9, J_{10}$.

Combining these results, we conclude that the Smith and Bao basis $$\{J_2,J_3,J_4,J_5,J_6,J_7,J_8,J_9,J_{10}\}$$
  is indeed an irreducible functional basis of $\D$.
\end{proof}

\section{The Mixed Function Basis is also an Irreducible Functional Basis}

We now consider the mixed functional basis $\{J_2,J_3,J_5,J_6,K_6,J_7,J_8,J_9,J_{10}\}$.  By Theorem \ref{t4.2}, each of $J_3, J_5, J_7$ and $J_9$ is not a function of the other eight invariants.   We may prove the following result.

\begin{Lemma} \label{l6.1}
For the mixed functional basis $\{J_2,J_3,J_5,J_6,K_6,J_7,J_8,J_9,J_{10}\}$, each of $J_2, K_6, J_8$ and $J_{10}$ is not a function of the other eight invariants.
\end{Lemma}
\begin{proof}   We consider two tensors, whose independent elements and associated values of invariants
are listed as
$$
\begin{tabular}{r|ccccccccc}
$(D_{1111}, \ldots, D_{2223})$            & $J_2$ & $J_3$ & $J_5$ &  $J_6$  & $K_6$  & $J_7$  & $J_8$  & $J_9$  & $J_{10}$ \\ \hline
$\D_1 :(\sqrt{2}, 0, 0, 0, 0, 0, 0, 0, 0)$        & 16 & 0 &0 &0 &1024 &0 &0 &0 &0    \\
$\D_2 :(0,\frac{2}{\sqrt[6]{31}}, 0, 0, 0, 0, 0, 0, 0)$ & $\frac{64}{\sqrt[3]{31}}$ &0 &0 &0 &1024 &0 &0 &0 &0
\end{tabular}
$$
Since $J_2(\D_1) \neq J_2(\D_2)$ and the corresponding values of $J_3, J_5, J_6, K_6, J_7, J_8, J_9$ and $J_{10}$ in $\D_1$
and $\D_2$ are equal, we conclude that $J_2$ is not a function of $J_3, J_5, J_6, K_6, J_7, J_8, J_9$ and $J_{10}$.

If $K_6$ is a function of the other right invariants, then the eight invariant set $$\{J_2,J_3,J_5,J_6,J_7,J_8,J_9,J_{10}\}$$ is a functional basis of $\D$.  This contradicts
Theorem \ref{t5.1}.
If $J_8$ is a function of the other eight invariants, then the eight invariant set $\{J_2,J_3,J_5,J_6,K_6,J_7,J_9,J_{10}\}$ is a functional basis of $\D$.  By the proof of Theorem \ref{t3.1}, $K_6$ is a linear combination of $J_2^3$, $J_2J_4$, $J_3^2$ and $J_6$.  This implies that the eight invariant set $\{J_2,J_3,J_4,J_5,J_6,J_7,J_9,J_{10}\}$ is also a functional basis of $\D$, and thus contradicts Theorem \ref{t5.1}.
If $J_{10}$ is a function of the other eight invariants, then the eight invariant set $\{J_2,J_3,J_5,J_6,K_6,J_7,J_8,J_9\}$ is a functional basis of $\D$.  Since $K_6$ is a linear combination of $J_2^3$, $J_2J_4$, $J_3^2$ and $J_6$, this implies that the eight invariant set $\{J_2,J_3,J_4,J_5,J_6,J_7,J_8,J_9\}$ is also a functional basis of $\D$, and thus contradicts Theorem \ref{t5.1}.
\end{proof}

Based on these results, we may prove the following theorem.
\begin{Theorem} \label{t6.2}
 The mixed functional basis $\{J_2,J_3,J_5,J_6,K_6,J_7,J_8,J_9,J_{10}\}$ is an irreducible functional basis of the fourth order three-dimensional symmetric and traceless tensor $\D$.
\end{Theorem}
\begin{proof}
By Theorem \ref{t4.2}, each of $J_3, J_5, J_7$ and $J_9$ is not a function of the other eight invariants.   By Lemma \ref{l6.1}, each of $J_2, K_6, J_8$ and $J_{10}$ is not a function of the other eight invariants. We only need to show that $J_6$ is not a function of $J_2, J_3, J_5, K_6, J_7, J_8, J_9, J_{10}$.
We adopt a similar tactic used in the proof of Theorem \ref{t5.1}, i.e., we try to find a solution of homogenous system of equations
$$ \left\{
\begin{array}{rcl}
J_2(D_{1123}, D_{1113}, D_{1223}, D_{2223}) &=&J_2(D_{1123}, D_{1113}, \hat{D}_{1223}, D_{2223}) \\
K_6(D_{1123}, D_{1113}, D_{1223}, D_{2223}) &=&K_6(D_{1123}, D_{1113}, \hat{D}_{1223}, D_{2223}) \\
J_8(D_{1123}, D_{1113}, D_{1223}, D_{2223}) &=&J_8(D_{1123}, D_{1113}, \hat{D}_{1223}, D_{2223}) \\
J_{10}(D_{1123}, D_{1113}, D_{1223}, D_{2223}) &=&J_{10}(D_{1123}, D_{1113}, \hat{D}_{1223}, D_{2223})
\end{array}
\right.
$$
such that $J_6(D_{1123}, D_{1113}, D_{1223}, D_{2223}) \neq J_6(D_{1123}, D_{1113}, \hat{D}_{1223}, D_{2223})$.
Note that $$J_3(\D)=J_5(\D)=J_7(\D)=J_9(\D)=0$$
if
$D_{1111}=D_{1112}= D_{1122}=D_{1222}=D_{2222}=0$.
Without loss of generality, we set $D_{1123}=1$. Using Wolfram Mathematica 11, we get an exact solution in the form of
$$
D_{1123}=1, \quad D_{1113} =-\sqrt{r}, $$
$$
D_{1223}= -\frac{1}{4} + \frac{1}{c_1 } \sqrt{a_0 +  a_1 r +  a_2 r^2 + \dots +  a_{34} r^{34} },
$$
$$
D_{2223}= -\frac{1}{c_2} \sqrt{r (b_0 +  b_1 r +  b_2 r^2 + \dots +  a_{34} r^{34}   )},
$$
$$
\hat{D}_{1223}= -\frac{1}{4} - \frac{1}{c_1} \sqrt{a_0 +  a_1 r +  a_2 r^2 + \dots +  a_{34} r^{34} },
$$
where $r \in (0.15, 0.18)$ is the unique solution of the polynomial equation
$$ -64109784657400987065159603764246576634003456000000000000$$ $$- 5369129429885302460742809297648749887571026247680000000000 x$$ $$-
 153201008141827442881381394920650784949762186471014400000000 x^2$$ $$ - 1656889511546590825873682312283824523248748098888269824000000 x^3$$ $$ +
 7369444722014488854141015192689977754200419787516235612160000 x^4$$ $$ + 344687045547898293313166391701750328724958377970996437975040000 x^5$$ $$ + 2192110181152338457993523061441529443647732547671602096373760000 x^6$$ $$ - 14968963396575624199399135040531244877145468364447279535325184000 x^7 $$ $$-
 234805849802658679752240863591832681103108824951597212530860032000 x^8 $$ $$- 422337927859225356000861814498346706911669673537592700848619520000 x^9$$ $$ +
 7028634526826762138513975906082250674889987023264347324441821593600 x^{10}$$ $$ + 36969946474677561620371874762436081248575723446112350394521421209600 x^{11}$$ $$ -
 42442848536296635464629775757996628679093460335317472356029802675200 x^{12}$$ $$ - 716662236647463544412755864040931894538401570933257995384451172234240 x^{13}$$ $$ -
 741853994396767828968888171006412749816485023762218636476738145861120 x^{14}$$ $$ + 7204018628316195277395849870383786198537765387905103672824498339991040 x^{15}$$ $$ +
 17348659913622400445975692756963677245348807841098400455241408150652416 x^{16}$$ $$ - 38182294348838453520904561075935755825444604439082277602081162833113216 x^{17}$$ $$ - 162620492217718142057174966825475718413039914424528022118121770126954832 x^{18}$$ $$ + 55889448851933401132472904264447517325518284609316690018472113411805952 x^{19}$$ $$ + 797081255784718719287897679123283874803228843631436387713517860839048467 x^{20}$$ $$ + 337528758551396077067549864406800157283411214325783486527445170676815643 x^{21}$$ $$ -
 2183993467513535801976109125792928890632360395728974728671707707490772642 x^{22}$$ $$ -
 1377517934978863146865638098333299244612192410198619044666080922573657842 x^{23}$$ $$ +
 4562544289015659637335110106732281009311625087430102047421541548985072359 x^{24}$$ $$ +
 3010505914931180973189105423110988828131288534529113892496394116850347815 x^{25}$$ $$ -
 6216715206473873649307593424153592668665711420866822835179265312439007400 x^{26}$$ $$ -
 2578957380026683259814608941321855563802049904499923250681850105712630000 x^{27}$$ $$ +
 5552422612712111836587440567052714712334011094843262604618429510827000000 x^{28}$$ $$ -
 526609414586276159483796972711762412239721424217621488944174883900000000 x^{29}$$ $$ -
 2377245057318921575984682390841590664129326443658943244650239690000000000 x^{30}$$ $$ +
 1562359154759369959310088265288232403322280944017751063536291750000000000 x^{31}$$ $$ -
 715784362177775248604714930110261684669270224278974153421200000000000000 x^{32}$$ $$ +
 134960627886741137293356075410432151464506032178953496560000000000000000 x^{33}$$ $$ -
 5205546624261647555731692508988955194694566406781056000000000000000000 x^{34}$$ $$ +
 120374337563506833410602387633027953028691383552000000000000000000000 x^{35}=0.$$
Except that $D_{1123}=1$, approximate digit values of $D_{1113}, D_{1223}, D_{2223}, \hat{D}_{1223}$ are as follows:
$$
D_{1113}=-0.405381, \  D_{1223}=0.67075, \  D_{2223}=1.12345  \text{  and }  \hat{D}_{1223}=-1.17075.
$$
Since $J_6(D_{1123}, D_{1113}, D_{1223}, D_{2223}) \neq J_6(D_{1123}, D_{1113}, \hat{D}_{1223}, D_{2223})$, it follows that $J_6$ is not a function of $J_2, J_3, J_5, K_6, J_7, J_8, J_9, J_{10}$.

Combining these results, we conclude that the mixed functional basis $$\{J_2,J_3,J_5,J_6,K_6,J_7,J_8,J_9,J_{10}\}$$ is an irreducible functional basis of $\D$.
\end{proof}

Apparently, there are many other irreducible functional bases which consist of polynomial invariants.  It is difficult to identify all of them.

%\begin{Corollary} \label{c6.3}
%For the Boehler, Kirillov and Onat basis $\{J_2, J_3, J_4, J_5, K_6, J_7, J_8, J_9, J_{10}\}$, each of $J_3, J_4, J_5, K_6, J_7, J_8, J_9$ and $J_{10}$ is not a function of the other eight %invariants.
%\end{Corollary}

\end{document}